\pdfoutput=1
\newif\ifFull
\Fulltrue
\ifFull
\documentclass[11pt]{article}  
\usepackage{fullpage}
\else
\documentclass[twoside,leqno,twocolumn]{article}  
\usepackage{ltexpprt}
\fi
\usepackage{graphicx}
\usepackage{latexsym}
\usepackage{amssymb}
\usepackage{amsmath}
\usepackage{url}

\newcommand{\R}{{\bf R}}

\ifFull
\newenvironment{proof}{\noindent{\bf Proof:}}{\hspace*{\fill}\rule{6pt}{6pt}\bigskip}
\newtheorem{theorem}{Theorem}
\newtheorem{lemma}[theorem]{Lemma}

\newtheorem{corollary}[theorem]{Corollary}

\fi

\begin{document}

\title{Linear-Time Algorithms for Geometric Graphs \\ 
with Sublinearly Many Edge Crossings}

\author{
David Eppstein\\
Dept.~of Computer Science \\
University of California, Irvine \\
\url{http://www.ics.uci.edu/~eppstein/}
\and 
Michael T. Goodrich \\
Dept.~of Computer Science \\ 
University of California, Irvine \\
\url{http://www.ics.uci.edu/~goodrich/}
\and
Darren Strash\\
Dept.~of Computer Science \\ 
University of California, Irvine \\
\url{http://www.ics.uci.edu/~dstrash/}
}

\date{}
\maketitle 

\begin{abstract} 
We provide linear-time algorithms 
for geometric graphs with sublinearly many edge crossings.
That is, we provide algorithms running in $O(n)$ time
on connected geometric graphs having $n$ vertices and $k$ crossings, where $k$
is smaller than $n$ by an iterated logarithmic factor.
Specific problems we study include Voronoi diagrams and single-source 
shortest paths.
Our algorithms all run in linear 
time in the standard comparison-based computational model;
hence, we make no assumptions about the distribution or bit
complexities of edge weights, nor do we utilize unusual bit-level
operations on memory words.
Instead, our algorithms are based on 
a \emph{planarization} method that ``zeroes in'' on edge crossings, together
with methods for extending planar separator decompositions 
to geometric graphs with sublinearly many crossings.
Incidentally, our planarization algorithm also 
solves an open computational geometry
problem of Chazelle for triangulating a self-intersecting polygonal
chain having $n$ segments and $k$ crossings 
in linear time, for the case when $k$ is 
sublinear in $n$ by an iterated logarithmic factor.
\end{abstract}

\section{Introduction}
A \emph{geometric graph}~\cite{ttgg} is an embedding of a graph
$G=(V,E)$ in $\R^2$ so that each vertex $v$ is associated with a unique 
point $p$ in $\R^2$ and each edge is ``drawn'' as a straight line segment
joining the points associated with its end vertices.
Moreover, the edges incident on each vertex $v$ are given in angular
order around $v$, so that faces in the embedding of
$G$ in $\R^2$ are well-defined (e.g., using the next-clockwise-edge ordering). 
Thus, we use the same notation and terminology to refer to
$G$ and its embedding.
If the edges in $G$ have no crossings, then $G$ is said to be a 
\emph{plane graph}, 
while graphs that admit realizations as plane graphs 
are \emph{planar graphs}~\cite{dett-gd-99,eg-sbdgt-70}.

Geometric graphs are natural abstractions of the geometric and
connectivity relationships that arise in a number of applications, 
including road networks, railroad
networks, and utility distribution grids, 
as well as
sewer lines and the physical connections defining the
Internet.
An example road network is shown in Figure~\ref{fig-ny}.

\begin{figure}[htb]
\begin{center}
\includegraphics[height=3in]{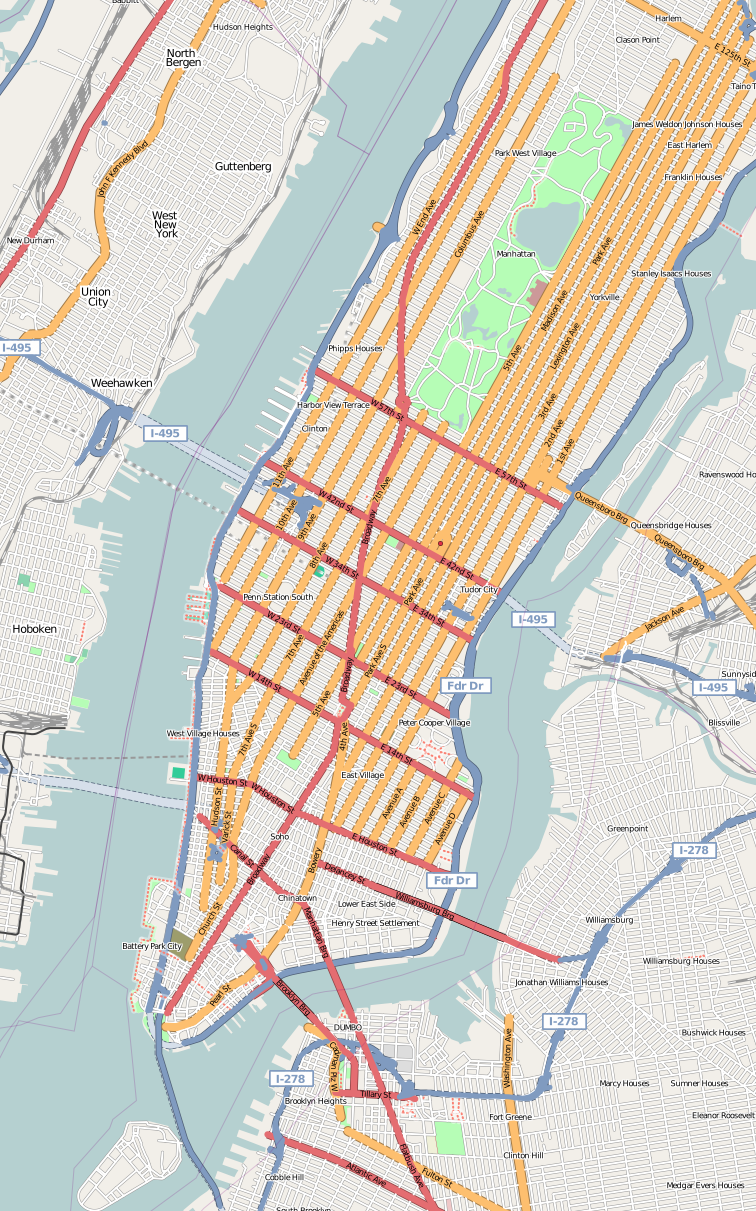} 
\end{center}
\caption{A portion of the road network 
surrounding the location of SODA 2009.
This image is from
\textsf{http://wiki.openstreetmap.org/},
under the Creative Commons attribution-share alike license.
}
\label{fig-ny}
\end{figure}

\begin{figure*}[hbt]
\centering\includegraphics[width=4.5in]{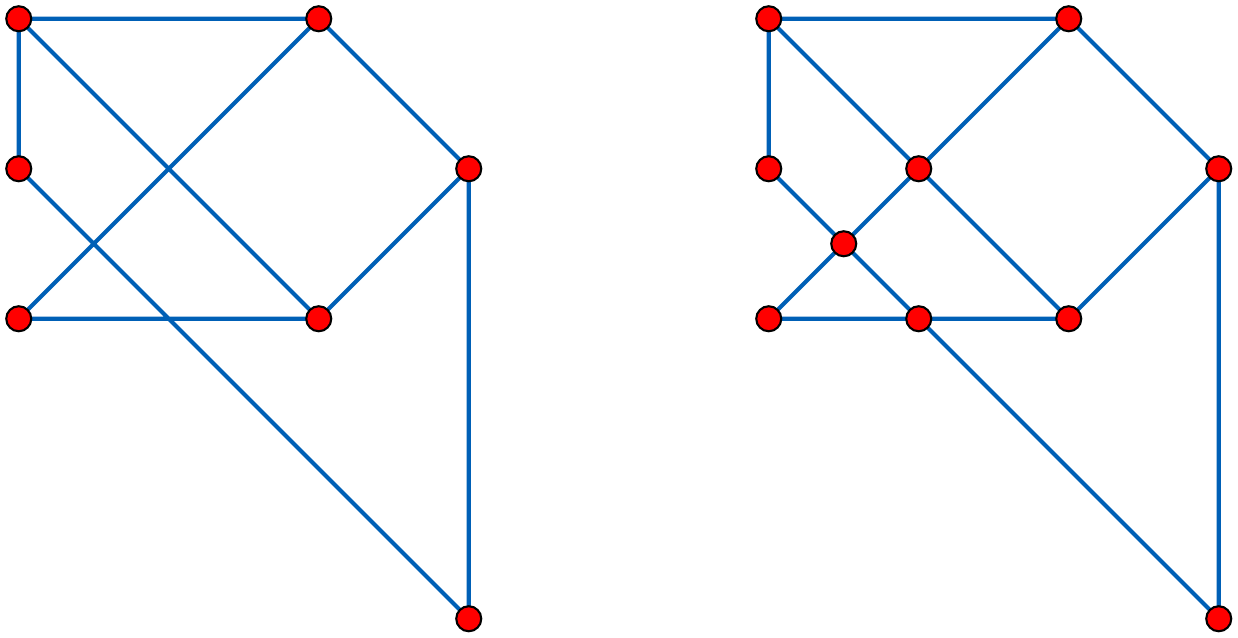}
\caption{A geometric graph and its planarization.}
\label{fig:planarization}
\end{figure*}

Although planar graphs and their plane graph realizations
have been studied extensively (e.g., see~\cite{t-pg-93}), 
real-world geometric graphs often contain edge crossings.
Recent experimental studies by the first two authors gives
empirical evidence that real-world road networks typically have
$\Theta(\sqrt{n})$ edge crossings, where $n$ is the number of
vertices~\cite{eg-snprntaal}.
Motivated by this real-world example, therefore, we are interested in
studying algorithms for connected geometric graphs that
have a sublinear number of edge crossings. However, we use a weaker restriction on the number of crossings than the bounds that our evidence suggests for road networks: here we are interested in $n$-vertex geometric graphs that
have at most $O(n/\log^{(c)} n)$ edge crossings, for some constant
$c$, where $\log^{(c)} n$ denotes the $c$-th iterated logarithm
function.
We refer to such geometric graphs as \emph{restrained} graphs.

Given an $n$-vertex geometric graph $G$, 
the \emph{planarization}\footnote{Our use of this term differs
	from its use in the graph drawing literature 
	(e.g., see~\cite{dett-gd-99}), where it refers to 
	the problem of removing a minimal number of 
	edges to make $G$ be planar.} 
of $G$ is the
graph $G'$ that is defined by the arrangement of the edges in $G$.
That is, as shown in Figure~\ref{fig:planarization}, we place a vertex in $G'$ for every vertex and pairwise edge
crossing in $G$, and we create an edge in $G'$ for every maximal edge segment
from $G$ that connects exactly two vertices in $G'$. Likewise, we
preserve the (clockwise/counterclockwise) ordering 
of edges around corresponding vertices 
in $G$ and $G'$, and we assume that intersection vertices in $G'$
similarly have their edges given in rotational order.
Thus, $G'$ is a plane graph having $n+k$ vertices,
where $k$ is the number of pairwise edge crossings among the edges in $G$.
By well-known properties of planar graphs (e.g., see~\cite[Prop.~2.1.6]{mt-gos-01}),
this implies that $G'$ has at most $3n+3k-6$ edges,
which in turn implies that $G$ has at most $3n+k-6$ edges.
Therefore, by restricting our attention to connected
geometric graphs with a sublinear number of edge crossings, we are,
by implication, focusing on connected geometric graphs that have $O(n)$ edges in their planarizations.

As mentioned above, a wealth of 
algorithms are known for planar graphs and plane graphs.
Indeed, many of these algorithms, for such problems as single-source
shortest paths and minimum spanning trees, run in $O(n)$ time.
Much less is known for non-planar geometric graphs, however, which
motivates our interest in such graphs in this paper.
Specifically, we are interested in the following problems for
connected, restrained geometric graphs:

\begin{itemize}
\item
The \emph{Voronoi diagram} problem,
which is also known as the \emph{post office} problem:
we are given a set $P$ of $k$ vertices in a geometric graph $G$
and asked to determine for every other vertex $v$ in $G$ the 
vertex in $P$ that is closest to $v$ according to the graph metric.
\item
The \emph{single-source shortest path} problem: we are given a
vertex $s$ and a geometric graph $G$ and asked to find the shortest
paths from $s$ to every other vertex in $G$.
\item
The \emph{polygon planarization} problem: given a geometric 
graph defining a non-simple polygon
$P$ having $n$ vertices, compute the arrangement of all the edges of
$P$, including vertices defined by the pairwise crossings 
of the edges in $P$.
\end{itemize}
In all these cases, we desire comparison-based
algorithms that require no additional 
assumptions regarding the distribution of edge weights, so that our algorithms
can apply to a wide variety of possible edge weights that may vary for different users, including
combinations of distance, 
travel time, toll charges, and subjective scores rating safety
and scenic interest~\cite{Epp-SJC-03}.

\subsection{Previous Related Work}
In the algorithms community, there has 
been considerable prior work on shortest
path algorithms for Euclidean graphs 
(e.g., see~\cite{gh-cspasm-05,hsww-csuts-05,%
kp-sppls-06,ss-hhhes-05,sv-speg-86,zn-spaeu-98}),
which are geometric graphs where edges are weighted by the
lengths of the corresponding line segments.
This prior work takes a decidedly
different approach than we take in this paper, however, in that it focuses on using special properties of the edge weights
that do not hold in the comparison model,
whereas we study road networks as geometric graphs with a sublinear 
number of edge crossings and we desire
linear-time algorithms that hold in the comparison model.

The specific problems for which we provide linear-time algorithms are
well known in the general algorithms and computational geometry
literatures.
For general graphs with $n$ vertices and $m$ edges, 
excellent work can be found on 
efficient algorithms in the comparison model,
including single-source
shortest paths~\cite{clrs-ia-01,gt-adfai-02,r-rrsss-97}, 
which can be found in $O(n\log n + m)$
time~\cite{ft-fhtui-87}, and
Voronoi diagrams~\cite{a-vdsfg-91,ak-vd-00},
whose graph-theoretic version
can be constructed in $O(n\log n + m)$ time~\cite{e-tgvda-00,m-afaas-88}. 
None of these algorithms run in linear time, even for
planar graphs. Linear-time algorithms for planar graphs
are known for single-source shortest paths~\cite{hkrs-fspap-97},
but these unfortunately do not immediately
translate into linear-time algorithms for non-planar geometric graphs.
In addition, there are a number of efficient
shortest-path algorithms that make assumptions about 
edge weights~\cite{g-saspp-93,gh-cspasm-05,m-ssspa-01,t-usspp-99};
hence, they are not applicable in the comparison model.

Chazelle~\cite{c-tsplt-91a} shows that any simple polygon can be triangulated
in $O(n)$ time and that this algorithm can be extended to
determine in $O(n)$ time, for any polygonal chain $P$, 
whether or not $P$ contains a self-intersection.
In addition, Chazelle posed as an open problem whether or not one can
compute the arrangement of a non-simple polygon in $O(n+k)$ time,
where $k$ is the number of pairwise edge crossings.
Clarkson, Cole, and Tarjan~\cite{cct-rpatd-92,cct-erpat-92} answer this question in the
affirmative for polygons with a super-linear number of crossings, as
they give a randomized algorithm that solves this problem in 
$O(n\log^* n + k)$ expected time.
There is, to our knowledge, no previous 
algorithm that solves Chazelle's open
problem, however, for non-simple polygons with a sublinear
number of edge crossings.

\subsection{Our Results}
In this paper, we 
provide the first linear-time algorithm for planarizing a non-planar
connected geometric graph having a number of pairwise edge crossings, $k$,
that is sublinear in the number of vertices, $n$, 
by an iterated logarithmic factor.
Specifically, we provide a randomized algorithm for planarizing 
geometric graphs in $O(n + k\log^{(c)} n)$ expected time, which is linear for
restrained geometric graphs.
Given such a planarization, we show how it can be used to help
construct an $O(\sqrt{n})$-separator decomposition of the original
graph in $O(n)$ time.
Furthermore, we discuss how such separator decompositions can then be
used to produce linear-time algorithms for a number of problems,
including Voronoi diagrams and single-source shortest paths.
We also show how our planarization algorithm can be used to solve
Chazelle's open problem of planarizing non-simple polygons in 
expected linear time
for polygons having a number of pairwise edge crossings that is
sublinear in $n$ by an iterated logarithmic factor.
Thus, combining this result with the polygon planarization
algorithm of Clarkson, Cole, and Tarjan~\cite{cct-rpatd-92,cct-erpat-92} provides a method for
planarizing an $n$-vertex polygon with $k$ edge crossings in
optimal $O(n + k)$ expected time, for all 
values of $k$ except those in the range $[n/\log^{(c)}n, n\log^* n]$.
Our result also implies that the convex hull of restrained non-simple
polygons can be
constructed in $O(n)$ expected time, which, to the best of our knowledge, was
also previously open.

Besides planar separator decompositions, which we discuss below,
another one of the techniques we use in this paper is a method for
constructing a $(1/r)$-cutting for the edges of a geometric graph,
$G$.
This is a proper triangulation\footnote{%
	A \emph{proper triangulation} is a 
	connected planar geometric graph such that every
	face is a triangle and every triangular face has exactly three
	vertices on its boundary.},
$T$, of the interior of 
the bounding box containing $G$ such
that any triangle $t$ in $T$ intersects at most $(1/r)n$
edges of $G$.
Using existing methods (e.g., see~\cite{a-gpia-91i,bs-ca-95,h-cpctp-00}), one can construct such
a $(1/r)$-cutting for $G$ in $O(n\log r + (r/n)k)$ time, where $n$ is
the number of vertices in $G$ and $k$ is the number of pairwise edge crossings. However, in our application such a bound would be nonlinear, as we require $r$ to be large.
We show, in Section~\ref{sec-cuttings}, that for connected geometric graphs such a cutting can be constructed in the faster expected time bound
$O(ns + (r/n)k)$, where $r\le n/\log^{(s)} n$.

\section{Separator Decompositions}

\begin{figure*}[htb]
\centering\includegraphics[width=4.5in]{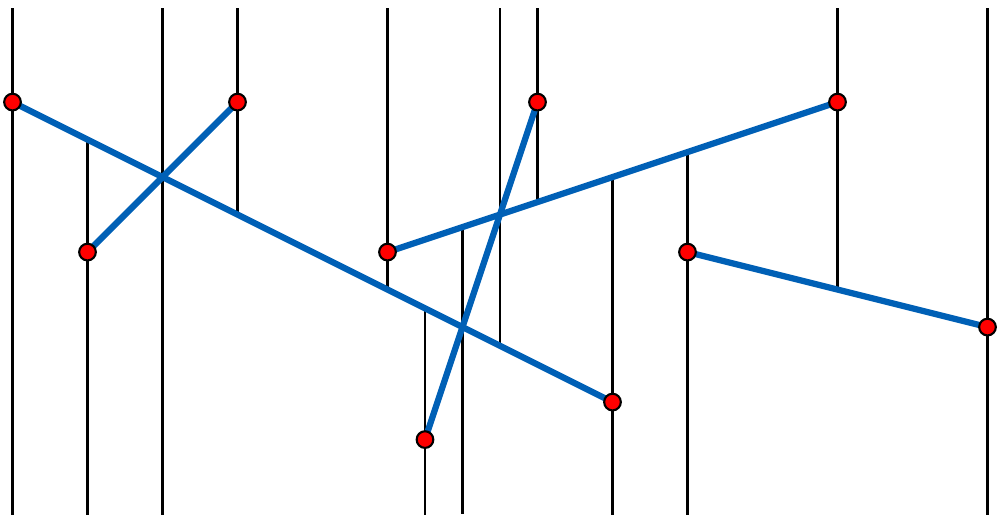}
\caption{Trapezoidal decomposition of a sampled subset of input graph edges.}
\label{fig:trapezoidalization}
\end{figure*}

One of the main ingredients
we use in our algorithms is the existence of small separators in
certain graph families (e.g., see~\cite{lt-stpg-79,m-fsscs-86}).
Several of the algorithms in this paper are based on the use of separators: we use them both as part of our algorithm for finding cuttings of geometric graphs, and later, once the graph has been planarized. Hence, we briefly review these tools here.

Given a graph $G=(V,E)$, a 
subset $W$ of $V$ is an \emph{$f(n)$-separator} if the removal of the
vertices in $W$ separates $G$ into two subgraphs $G_1$ and $G_2$,
each containing at most $\delta n$ vertices, for some constant
$0<\delta < 1$.
It is well known that planar graphs have $O(\sqrt{n})$-separators
with $\delta=2/3$, and that such separators can be constructed in
$O(n)$ time~\cite{lt-stpg-79}.
Such separators are typically used in divide-and-conquer algorithms,
which involve finding a separator, recursively solving the problem in
the two separated subgraphs, and then merging the solutions together.
If the merge and divide steps can be solved in $o(n)$ time, however,
it is useful to have the entire recursive separator decomposition
computed in advance; for otherwise there is no way to beat an
$O(n\log n)$ time bound. 
Such a \emph{separator decomposition} defines a binary tree $B$,
such that the root of $B$ is associated with the $f(n)$-separator for
$G$ and the subtrees of this root are defined recursively for the
graphs $G_1$ and $G_2$, respectively.

Previous work on separators includes the seminal contribution of
Lipton and Tarjan~\cite{lt-stpg-79}, who show that $O(\sqrt{n})$-sized
separators exist for $n$-vertex planar graphs and these can be computed
in $O(n)$ time.
Goodrich~\cite{g-psppt-95} shows that recursive
$O(\sqrt{n})$-separator decompositions can be constructed for planar
graphs in $O(n)$ time.
A related concept is that of geometric separators, which use
geometric objects to define separators in graphs defined by systems of intersecting disks
(e.g., see~\cite{ar-dsa-93,mttv-gsfem-95,mttv-sspnn-97,st-dpps-96}).
Eppstein
{\it et al.}~\cite{emt-dltag-93} provide a linear-time construction
algorithm for geometric separators which translates into an $O(n\log n)$ recursive separator
decomposition algorithm.

Because restrained graphs are not planar, the result of Goodrich does not immediately apply. However, it can be applied once we have planarized the graph, and it can also be applied to planar structures formed from subsets of the graph, such as the one we describe in the next section.

\section{Trapezoidal Decomposition of a Sample}

Suppose we are given a geometric graph $G$ having $n$ vertices
and $k$ pairwise intersections among its edges.
In this section, we describe our algorithm for constructing a
trapezoidal decomposition of a random sample of the edges of $G$.
That is, given the sample of edges, we construct the arrangement of these edges together with a set of vertical line segments through each edge endpoint and crossing, where each such segment is maximal with respect to the property of not crossing any other sampled edge, as shown below.
(See Figure~\ref{fig:trapezoidalization}.)
Our method is parameterized by $s$ where $r\le n/\log^{(s)} n$, and the sample probability is inversely proportional to $\log^{(s)} n$. We will later 
show how to refine this sample so that we can
produce a cutting and then a planarization of $G$.

This first step of
our algorithm is essentially the same as performing $s$ levels of the
Clarkson, Cole, and Tarjan algorithm, except that their
method is for polygonal chains, whereas ours is for geometric graphs. Thus, we describe it at a high level.

Our algorithm begins with a trivial trapezoidal decomposition $T_0$ containing a single trapezoid
that encloses all of $G$. Call this trapezoid $t$. Let $C(t)=E$ be the \emph{conflict list} for $t$, that is,
the set of edges from $G$ that intersect the interior of $t$.
Then, for $i=1$ to $s$, we perform the following computation.

\begin{enumerate}
\item
Find a random sample $S_i$ of size $n/\log^{(i)} n$, 
of the edges in $G$,
and for each trapezoid $t$ in $T_{i-1}$, use the Bentley-Ottmann 
algorithm~\cite{bo-arcgi-79} to
construct the trapezoidal decomposition of the arrangement of the segments in
$C(t)\cap S_i$. Once all these trapezoidal decompositions are constructed,
merge them together to create a single trapezoidal decomposition, $T_i$, for
the segments in $S_i$. To be consistent with Clarkson, Cole, and Tarjan, we choose the samples such that $S_1 \subset S_{2} \subset \cdots \subset S_s$.
\item
Perform a depth-first traversal of $G$, while keeping track of the trapezoids in
the trapezoidal decomposition that are intersected during the 
walk, so as to determine, for each trapezoid $t$ in $T_i$, the set $C(t)$. Since the geometric graph is connected, we never have to restart the depth-first traversal from a node whose location we do not already know. We can therefore use the arrangement of the sampled line segments to keep track of the intersected trapezoids at each step of the traversal. Thus we eliminate the need for time-consuming point-location data structure lookups.
\end{enumerate}

Let $T=T_s$ be the resulting final trapezoidal decomposition we get
from this computation, and let $S=S_s$ be the final random sample.
Using the framework established by Clarkson and Shor~\cite{cs-arscg-89} for
randomized divide-and-conquer algorithms, such as this, we can show
that
\begin{equation}
\label{eq-1}
E\left( |T| \right) = O\left(r+\left(\frac{r}{n}\right)^2 k\right)
\end{equation}
and
\begin{equation}
\label{eq-2}
E\left( \sum_{t\in T} |C(t)| \right) = 
O\left( n+\left( \frac{r}{n}\right) k\right).
\end{equation}
In particular, 
Equation~(\ref{eq-1}) is from their Lemma~4.1
and
Equation~(\ref{eq-2}) follows from their Corollary~4.4.  The number of steps in the depth-first traversal is proportional to the total size of the conflict lists of the input geometric graph with the trapezoidal decomposition, which as we have seen above is small. A step from one trapezoid to a horizontally adjacent trapezoid may be accomplished in constant time, but a single trapezoid may have a non-constant number of neighbors above and below it, causing steps in those directions to take longer. But as Clarkson, Cole, and Tarjan show, the sum over all trapezoids of the conflict list size of the trapezoid multiplied by its number of neighbors remains linear in expectation, and this sum bounds the time to step vertically from one trapezoid to another using a sequential search along the trapezoid boundary to find the neighboring trapezoid. Therefore, we have the following preliminary result:

\begin{lemma}
Given a connected geometric graph $G$ with $n$ edges and $k$ pairwise edge crossings, and a
parameter $s$, we can in expected 
time $O(ns+(r/n)k)$ find a random sample of $r=O(n/\log^{(s)} n)$ 
edges from $G$, the trapezoidal decomposition induced by the sample,
and the set of edges of $G$ crossing each trapezoid of the sample.
\end{lemma}

\section{Cuttings}
\label{sec-cuttings}

At this stage we take a detour from the Clarkson, Cole, and Tarjan algorithm.
For each trapezoid $g$ in $T$, let $\alpha_g=|C(g)|r/n$.
That is, $\alpha_g$ is the degree of excess that the conflict list
for $g$ has beyond what we would like for a $(1/r)$-cutting.
For each trapezoid $t$ with
$\alpha_t>1$, we form a random sample, $R_t$, of $C(t)$ of size 
$2b\alpha_t \log \alpha_t$, where $b$ is the constant $K_{\max}$ from 
Corollary 4.4 of Clarkson-Shor~\cite{cs-arscg-89}.
We then form the trapezoidal decomposition, $T_t$ of 
the arrangement of the segments in $R_t$ using 
any quadratic-time line segment arrangement algorithm~\cite{am-dplr-91,bdsty-arsol-92,ce-oails-92,ejps-pcoas-91}.
Thus, by Corollary 4.4 from Clarkson-Shor~\cite{cs-arscg-89},
the maximum size of any conflict list
of a trapezoid in $T_t$ is expected to be 
less than
\begin{eqnarray*}
\left(\frac{|C(t)|}{|R_t|}\right)\log |R_t| &=& 
\left(\frac{n}{r}\right)
      \left(\frac{1}{\log \alpha_t^2}\right)\log (2\alpha_t\log
      \alpha_t) \\
      &\le& \frac{n}{r},
\end{eqnarray*}
for $\alpha_t\ge 4$.
Thus, we can repeat the above algorithm an expected constant number
of times until we have this condition satisfied, which gives us one
of the crucial properties of a $(1/r)$-cutting: namely, that each
cell intersects at most $(n/r)$ edges of $G$.

In addition,
the number of new trapezoids created inside $t$, as well as the
running time for creating the trapezoidal diagram $T_t$, is certainly at most
$O(|R_t|^2)$, which is
$O(\alpha_t^2 \log^2 \alpha_t)$.
More importantly, we have the following:

\begin{lemma}
Given the above construction applied to each trapezoid $t$ in $T$,
then
\[
E\left( \sum_{t\in T} \alpha_t^2 \log^2 \alpha_t\right) 
= O\left(r+\left(\frac{r}{n}\right)^2 k\right).
\]
\end{lemma}
\begin{proof}
Our proof is based on an application of
Theorem~3.6 from the Clarkson-Shor framework. 
To apply this theorem, we bound
\[
E\left( \sum_{t\in T} \alpha_t^2 \log^2 \alpha_t\right)
\]
by bounding the term,
$\alpha_t^2 \log^2 \alpha_t$,
by 
\[
 W\left({{|C(t)|} \choose c }\right), 
\]
where $W$ is a positive concave function on ${\bf R}^+$ and $c$ is a constant.
Here, for the sake of an upper bound, we take $c=3$ and we define
\[
W(x) = \left(\frac{x^{1/3}}{N}\right)^2 \log^2 \frac{x^{1/3}+N}{N} ,
\]
where $N=n/r$.
Finally, to apply Theorem~3.6 from \cite{cs-arscg-89}, 
we need to observe that the number of trapezoids in $T$ that have a conflict list size
at most $c$ is proportional to the number
of trapezoids in $T$ that have a conflict list size at least $0$, which is $|T|$.
To see this, note that we can extend the vertical edges of any trapezoid in
$T$ in at most $O(1)$ ways until it hits $i=1,2,3$ other edges of the
random sample, $S$,
at which point we can extend this trapezoid horizontally in $O(1)$
ways until we hit $3$ segments in total.
Therefore, 
by Theorem~3.6 from \cite{cs-arscg-89}, 
\[
E\left( \sum_{t\in T} \alpha_t^2 \log^2 \alpha_t\right)
\]
is 
\[
O\left(r+\left(\frac{r}{n}\right)^2 k\right).
\]
\end{proof}

Thus, our refined 
trapezoidal decomposition, $T'$, will 
have size proportional to $|T|$.
It is still not quite a $(1/r)$-cutting, however, as it is not a
proper triangulation. Indeed, some trapezoids may
have many more than $4$ vertices on their boundaries 
(see Figure~\ref{fig: highdegtrap}).

\begin{figure}[hbt]
\centering\includegraphics[width=3in]{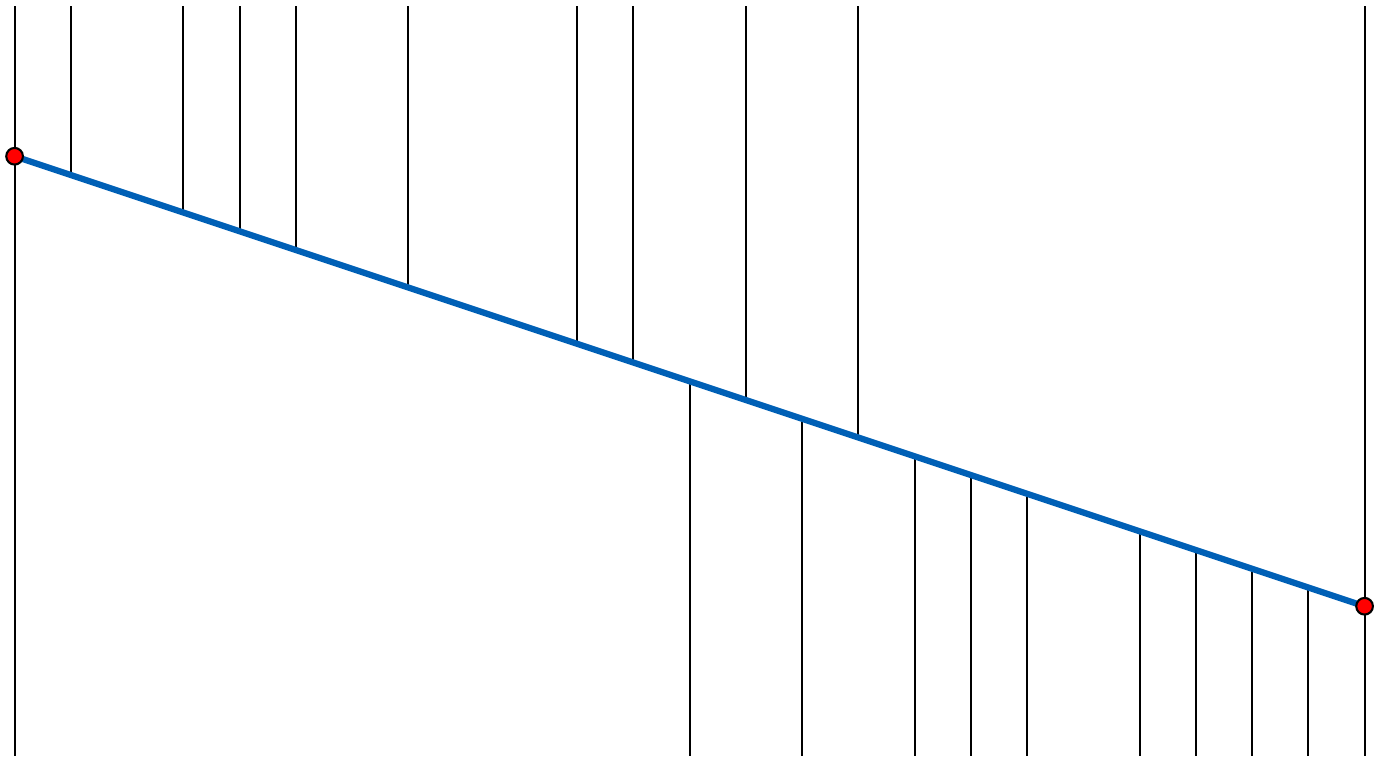}
\caption{Many trapezoids may be adjacent to another trapezoid along its top or bottom edges.}
\label{fig: highdegtrap}
\end{figure}

To refine $T'$ into a proper triangulation, we borrow an idea 
from the fractional cascading framework of Chazelle and
Guibas~\cite{cg-fc1ds-86} to first refine $T'$ into
a trapezoidal decomposition such that each trapezoid has $O(1)$
vertices on its boundary, while keeping the total number of trapezoids to be
$O(|T'|)$, which is expected to be 
\[
 O\left(r+\left(\frac{r}{n}\right)^2 k\right).
\]
By triangulating the interior of each 
such trapezoid, we will get a $(1/r)$-cutting whose size is still 
$O(|T'|)$.
(See Figure~\ref{fig:cascade}.)

\begin{figure}[hbt]
\centering\includegraphics[width=3in]{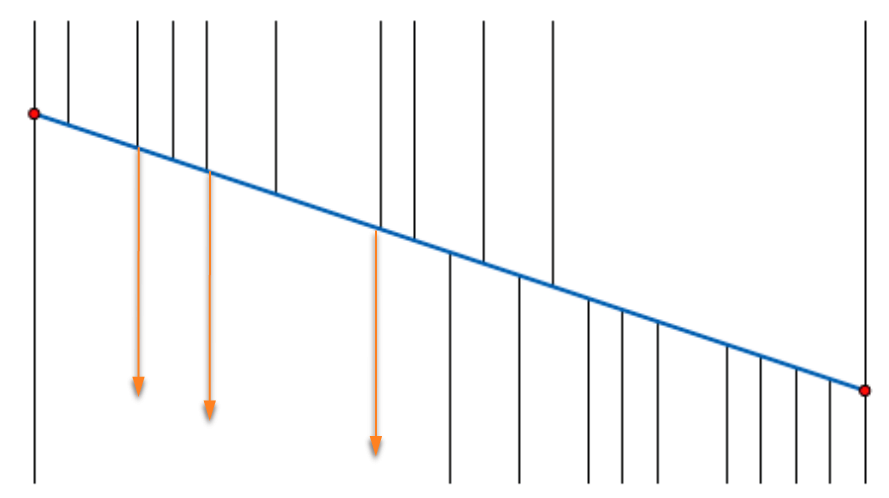}
\caption{The cascading of trapezoidal rays.}
\label{fig:cascade}
\end{figure}

Construct the graph-theoretic planar dual $U$ to $T'$, and note that
we can direct the edges of $U$ so as to define four directed-acyclic
graphs, which respectively define the partial orders ``below,''
``above,'' ``left-of,'' and ``right-of'' among the trapezoids.
Without loss of generality,
let us direct $U$ according to the ``below''
relation, perform a topological sort, and process the trapezoids of
$T'$ from top to bottom according to this ordering.
When processing a trapezoid, $t$, we assume inductively that we have
determined the ordered list of vertices $V_t=(v_1,v_2,\ldots,v_j)$ on
$t$'s upper edge, which are bottom vertices of trapezoids above $t$.
To process $t$ we choose every other vertex, $v_{2i}$,
in $V_t$ and extend a vertical segment from $v_{2i}$ to the bottom of
$t$ to split $t$ in two for each such $v_{2i}$.
Doing this for every other vertex in $V_t$, therefore, splits $t$ and
increases the number of trapezoids by $\lfloor |V_t|/2\rfloor$.
We then repeat this computation by considering the new set of
trapezoids according to the ``above'' relation, from bottom to top.
Next, we do a similar computation for the ``left-of'' and
``right-of'' relations (except that now we extend segments parallel
to the top or bottom edges of our trapezoid in a way that partitions
its interior into non-crossing trapezoids).
When we have completed this last scan of the trapezoids, we will have
created a trapezoidal decomposition such that each trapezoid has
$O(1)$ vertices on its edges.
More importantly, we also have the following:

\begin{lemma}
The total number of trapezoids created by the above refinement
process is $O(|T|)$, which has expected value $O(r+(r/n)^2k)$.
\end{lemma}
\begin{proof}
We have already established that
$E(|T|)$ is $O(r+(r/n)^2k)$ and that 
$E(|T'|)$ is $O(E(|T|))$.
So we have yet to show that the number of new trapezoids created
during any of our splitting processes is $O(|T'|)$.
We do this by an accounting argument.
Without loss of generality
, consider the processing according to the ``below'' relation.
Assume, for the sake of our analysis, that, at the beginning
of our computation, we give each vertical edge in our trapezoidal 
decomposition \$2 and we require every vertical edge at the end of
the process to have at least \$1.
When we extend a vertical ray from an even numbered vertex
$v_{2i}$ at the top of a trapezoid, $t$ we can assume inductively that
the vertical edge above $v_{2i}$ has \$2, as does the vertical edge
directly to the left of this edge (which hits $t$ at vertex
$v_{2i-1}$).
Let us take \$1 from this vertical edge and from the one that hits
$t$ at $v_{2i}$, which leaves \$1 at each of those edges,
and use the \$2 to pay for the new vertical edge that we then extend
through $t$.
Therefore, since the two vertical edges we just took money from will
not be processed again, we can process each trapezoid and pay for
every action, while keeping \$1 for each trapezoid in our refined
trapezoidal decomposition.
Repeating this accounting argument for the ``above,'' ``left-of,''
and ``right-of'' relations completes the proof.
\end{proof}

Given a trapezoidal diagram having $O(1)$ vertices on the boundary of
each trapezoid, and each trapezoid intersecting at most $(n/r)$ edges
of our geometric graph $G$ we can easily triangulate 
each trapezoidal face in this diagram to turn it into a
$(1/r)$-cutting with a number of triangles that is proportional to
the number of trapezoids.
(See Figure~\ref{fig:trias}.)

\begin{figure}[hbt]
\centering\includegraphics[width=3in]{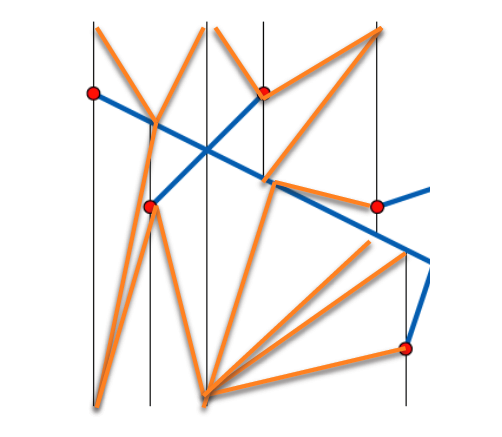}
\caption{The triangulation step.}
\label{fig:trias}
\end{figure}

Thus, putting all the pieces together, we get the following.

\begin{theorem}
\label{thm-cutting}
Given a connected geometric graph $G$ having $n$ vertices and $k$ pairwise edge
crossings, one can construct a $(1/r)$-cutting for
the edges of $G$ of expected size $O(r+(r/n)^2k)$ in expected time 
$O(ns+(r/n)k)$, for $r\le n/\log^{(s)} n$.
\end{theorem}

Taking $s$ as a constant gives us such a $(1/r)$-cutting 
of expected size $O(r+(r/n)^2k)$ in expected time 
$O(n+(r/n)k)$, and taking $s=\log^* n$ gives us
a $(1/r)$-cutting 
of the same expected size (but with a potentially larger $r$) in
expected time $O(n\log^* n + (r/n)k)$, for any $r\le n$.
Since, in our applications involving restrained geometric graphs,
$k$ is sublinear in $n$ by an iterated logarithmic factor, we will be
taking $s$ to be a constant.

\section{Planarization}
In this section, we describe how to planarize a connected
geometric graph $G$ having $n$ vertices and $k$ edge crossings.
We begin by using the method
of Theorem~\ref{thm-cutting} to construct
a $(1/r)$-cutting, $C$, of the edges of $G$ of expected size $O(r+(r/n)^2k)$
in expected time $O(n+(r/n)k)$, where $r = n/\log^{(c+1)} n$, for a
fixed constant $c\ge 1$.
We then do a depth-first search of $G$, keeping track of the
triangles we cross in $C$ as we go, to compute, for each triangle $t$
in $C$, the set, $C(t)$,
of at most $(n/r)$ edges of $G$ that intersect $t$.
This takes $O(|C|n/r)$ time, which has expectation $O(n+(r/n)k)$.

We then apply Goodrich's separator decomposition algorithm~\cite{g-psppt-95}
to construct an $O(\sqrt{|D|})$-separator decomposition of the
graph-theoretic dual, $D$, to $C$.
Rather than taking this decomposition all the way to the point where
we would have subgraphs of $D$ of constant size, however, we stop
when subgraphs have size $O(\log^2 (n/r))$; hence, have separators of
size $O(\log (n/r))$.
Since $C$ is a triangulation, $D$ has degree $3$; hence, any
vertex separator for $D$ of size $g$ also gives us an edge separator
for $D$ of size at most $3g$.
Moreover, each edge of $D$ corresponds to a triangle edge in $C$,
which in turn crosses at most $(n/r)$ edges of $G$.
For each separator $H$ in our decomposition,
therefore, we can sort the edges of $G$ that cross each boundary of
a triangle in the separator 
in time $O((n/r)\log (n/r))$ time.
There are $O(|D|/\log^2 (n/r))$ nodes at this level of the separator
decomposition tree; hence, there are 
$O(|D|/\log^2 (n/r))\times O(\log (n/r)) = O(|D|/\log (n/r))$
triangles involved.
Thus, the total time for all these sorts is 
$O(|D|(n/r)) = O(n+k)$.

After performing all these sorts of edges on the boundaries 
of triangles in our separators, we can imagine that we have used
these boundaries to cut $G$ into $O(|D|/\log (n/r))$ regions
(including each triangle in one of our separators), such
that the edges of $G$ intersecting each region boundary 
are given in sorted order.
(See Figure~\ref{fig:regions}.)

\begin{figure}[hbt]
\centering\includegraphics[width=3in]{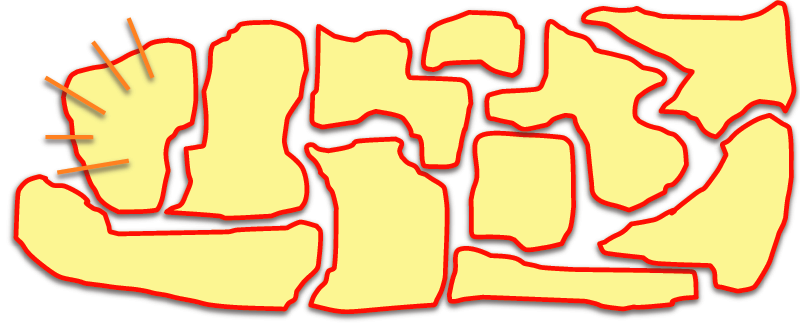}
\caption{Illustrating the regions and their boundary edges.}
\label{fig:regions}
\end{figure}

The total size of each subgraph is $O((n/r)\log^2 (n/r))$.
Moreover, the boundaries of these regions form a planar subdivision.
Thus, we have just subdivided our geometric graph $G$ into 
$O(|D|/\log (n/r))$ disjoint geometric graphs.
In other words,
all $k$ edge crossings in $G$ have been isolated
into these small subgraphs.

For each subgraph $G_i$, use Chazelle's algorithm~\cite{c-tsplt-91a} to test
if all the faces of $G_i$ are simple in $O(|G_i|)$ time.
If all the faces of $G_i$ are in fact simple, then $G_i$ clearly
contains no edge crossings.
Thus, we can identify each small subgraph in this partition that
contains an intersection in time $O(|C|(n/r) + |G|)$, which has
expectation $O(n+k)$.

Clearly, 
there are at most $k$ such subgraphs that contain edge crossings.
We complete our planarization algorithm, therefore, by running the
Bentley-Ottmann algorithm~\cite{bo-arcgi-79}
for each subgraph of $G$ that is identified
as having at least one edge crossing.
The time for each such invocation of the Bentley-Ottmann algorithm is
$O((n/r)\log^3 (n/r) + k'\log (n/r))$, where $k'\ge 1$ is the 
number of edge crossings found.
Summing this over $k$ regions implies that the total time needed to
complete the planarization of $G$ 
is $O(k(n/r)\log^3 (n/r))$.
Substituting for $r$, we see that this time is
$O(k \log^{(c+1)} n \log^3 \log^{(c+1)} n)$,
which is $O(k \log^{(c)} n )$.
Therefore, we have the following:

\begin{theorem}
\label{thm-planarization}
Suppose one is
given a connected geometric graph $G$ with $n$ vertices and
$k$ edge crossings, 
together with a $(1/r)$-cutting of the edges of $G$ of size
$O(r+(r/n)^2k)$, for $r=n/\log^{(c+1)} n$.
Then one can construct a planarization of $G$ (and the
trapezoidal decomposition of the arrangement of $G$'s edges), in
time $O(n+k\log^{(c)} n)$.
\end{theorem}

Combining this result with Theorem~\ref{thm-cutting}, we
get the following corollary.

\begin{corollary}
\label{cor-planarization}
Given a connected geometric graph $G$ having $n$ vertices and $k$ pairwise edge
crossings, 
one can construct a planarization of $G$ in
expected time $O(n+k\log^{(c)} n)$.
\end{corollary}

\section{Applications}
In this section, we provide a number of applications of the above algorithms.

\subsection{Separator Decompositions of Restrained Geometric Graphs}

The algorithms in this section are based on the use of separators.
As mentioned above, 
the separator-decomposition algorithm of Goodrich~\cite{g-psppt-95} 
applies only to planar graphs.
Nevertheless, 
given the tool of geometric graph planarization,
we can adapt Goodrich's result to restrained geometric graphs in a fairly
straightforward manner.
Given a restrained 
geometric graph $G$, we planarize it using the algorithm above, creating the
planar graph $G'$.
As observed above, 
$G'$ has total size $O(n)$. 
Thus, we can use the result of Goodrich~\cite{g-psppt-95} 
to compute a recursive
$O(\sqrt{n})$-separator decomposition of $G'$ in $O(n)$ time.
We convert this separator decomposition into a 
$O(\sqrt{n})$-separator for $G$ by the following transformation.
For each node $v$ in a separator $W$ of $G'$
at a node $w$ in the separator decomposition tree $B$,
we do the following:
\begin{itemize}
\item
If $v$ is also a vertex in $G$, then we
add $v$ to the separator for $G$ corresponding to $w$, provided $v$ is not
already a member of a separator associated with an ancestor of $w$.
\item
If $v$ is an intersection point in $G'$, between edges $(a,b)$ and $(c,d)$ 
in $G$, then we add each of $a$, $b$, $c$, and $d$
to the separator for $G$ corresponding to $w$, provided it is not
already a member of a separator associated with an ancestor of $w$.
\end{itemize}
This gives us the following:

\begin{theorem}
\label{thm-deterministic}
Suppose we are given an $n$-vertex geometric graph $G$ 
and its planarization, $G'$, which is of size $O(n)$.
Then we can construct a recursive $O(\sqrt{n})$-separator decomposition of
$G$ in $O(n)$ time, for $\delta=2/3$.
\end{theorem}

\subsection{Single-Source Shortest Paths and Voronoi Diagrams}
Given an $n$-vertex
bounded-degree graph $G$ and a recursive $O(\sqrt{n})$-separator 
decomposition for $G$,
Henzinger {\it et al.}~\cite{hkrs-fspap-97} show that one can compute shortest
paths from a single source $s$ in $G$ to all other vertices in $G$ in $O(n)$
time.
Using the separator decomposition algorithms presented above, then,
we can show that
their algorithm applies to restrained geometric graphs, even ones that do not have
bounded degree, by a simple transformation that replaces high-degree
vertices with bounded-degree trees of zero-weight edges.

Suppose we are given $K$ distinguished vertices in an $n$-vertex
restrained geometric graph $G$ and we
wish to construct 
the \emph{Voronoi diagram} of $G$, which is a labeling of each
vertex $v$ of $G$ with the name of the distinguished vertex closest
to $v$.
As before, by replacing high degree vertices with bounded-degree 
trees of zero-weight edges we can assume without loss of generality 
that $G$ has constant degree.
In this case, we construct a recursive $O(\sqrt{n})$-separator
decomposition of $G$ using one of the algorithms of the previous
section.
Let $B$ be the recursion tree and let us label each vertex $v$ in $G$
with the internal node $w$ in $B$ where $v$ is added to the separator
or with the leaf $w$ in $B$ corresponding to a set containing $v$ where
we stopped the recursion (because the set's size was below our
stopping threshold).
Given this labeling, we can trace out the subtree $B'$ of $B$ 
that consists of the union of paths from the root of $B$ to the
distinguished nodes in $G$ in $O(n)$ time.
Let us now assign each edge in $B'$ to have weight $0$ and let us add
$B'$ to $G$ to create a larger graph $G'$.
Note that if we add each internal node $v$ in $B'$ to the separator
associated with node $v$ in $B$, then we get a recursive 
$O(\sqrt{n})$-separator decomposition for $G'$, for each separator in
the original decomposition increases by at most one vertex.
Thus, we can apply the algorithm of Henzinger {\it et al.}~\cite{hkrs-fspap-97}
to compute the shortest paths in $G'$ from the root of $B'$ to every
other vertex in $G'$ in $O(n)$ time.
Moreover, since the edges of $G'$ corresponding to edges of $B'$ have
weight $0$, this shortest path computation will give us the Voronoi
diagram for $G$.
Therefore, we have the following:

\begin{theorem}
\label{thm-vd}
Given a connected $n$-vertex restrained graph $G$, together with
its planarization,
one can compute shortest paths from any vertex
$s$ or the Voronoi diagram defined by any set of $K$ vertices in $G$
in $O(n)$ time.
\end{theorem}

Incidentally, the above approach also implies a linear-time Voronoi diagram 
construction algorithm for planar graphs, which was not previously known.

\section{Conclusions and Future Work}
We have provided linear-time algorithms for a number of problems
on connected restrained geometric graphs, which includes real-world
road networks. Our results allow for linear-time trapezoidalization, triangulation, and planarization of geometric graphs except for the very narrow range of the number of crossings for which neither our algorithm nor the previous $O(n\log^* n+k)$ algorithm is linear.
In addition, our methods imply linear-time algorithms for other problems on
such graphs as well. For example, one can use our algorithm to planarize a
restrained non-simple polygon and then construct its convex hull in linear
time by computing the convex
hull of the outer face of our planarization (e.g., by an algorithm
from~\cite{gy-fchsp-83,l-fchsp-83}).
There are a number of interesting open problems and future research directions 
raised by this paper, including: 
\begin{itemize}
\item
Can one close the $\log^{(c)} n$ gap on values of $k$ 
that admit optimal 
solutions to Chazelle's
open problem of
computing a trapezoidal decomposition of an $n$-vertex non-simple polygon in
$O(n+k)$ time, where $k$ is the number of its edge crossings?

\item Can we planarize restrained geometric graphs deterministically in linear time? Such a result would allow us to apply separator-based divide and conquer techniques for minimum spanning trees~\cite{EppGalIta-JCSS-96} to construct them in linear time for this family of graphs. Known linear-time minimum spanning tree algorithms for arbitrary graphs require randomization~\cite{kkt-rltamst-95}, and known deterministic algorithms for this problem are superlinear~\cite{c-amsta-00}, although deterministic linear-time algorithms are known for planar graphs and minor-closed graph families~\cite{ct-fmst-76,Epp-HCG-00,m-tltam-04}.
\end{itemize}

\subsection*{Acknowledgment}
We would like to thank Bernard Chazelle for several helpful 
discussions regarding possible approaches to solving his open problem
involving non-simple polygons.
This research was supported in part by the National Science Foundation, under
grants 0724806, 0713046, and 0830403,
and the Office of Naval Research, under
MURI Award number N00014-08-1-1015.
A preliminary version of this paper appeared in the ACM-SIAM
Symposium on Discrete Algorithms (SODA) as~\cite{egs-ltagg-09}.

\ifFull
\bibliographystyle{abbrv}
\bibliography{extra,geom,roads,goodrich}

\else

\fi
\end{document}